\documentclass[runningheads, orivec]{llncs}

\usepackage{amsthm}
\usepackage[T1]{fontenc}
\usepackage{listings}
\usepackage[frozencache, cachedir=minted-cache]{minted2}
\usepackage{amsmath}
\usepackage{amssymb}
\usepackage{hyperref}
\usepackage{wrapfig}
\usepackage{graphicx}
\usepackage{xspace}
\usepackage[final,commandnameprefix=always,deletedmarkup=xout]{changes}
\usepackage{orcidlink}
\usepackage{bbding}
\usepackage{siunitx}
\usepackage[title]{appendix}
\usepackage{acronym}
\usepackage{pgf}
\usepackage{lmodern}
\usepackage{mathtools}
\usepackage{ifsym}
\usepackage{import}

\everymath=\expandafter{\the\everymath\displaystyle}
\makeatletter\@ifpackageloaded{underscore}{}{\usepackage[strings]{underscore}}\makeatother

\definechangesauthor[name=Alistair, color=blue]{AS}
\definechangesauthor[name=Fleur, color=green]{FC}
\newcommand{\listingsize}{\small}
\newcommand{\network}{\mathcal{N}}

\newcommand{\vehicle}{Vehicle\xspace}
\newcommand{\rocq}{Rocq\xspace}
\newcommand{\csafe}{\ensuremath{C_{\mathit{safe}}}}
\newcommand{\aref}[1]{\hyperref[#1]{Appendix~\ref*{#1}}}
\newcommand{\ct}{\ensuremath{C_\Sigma}}

\newcommand{\concat}{\mathbin{{+}\mspace{-8mu}{+}}}

\begin{document}
\newacro{PK}{pharmacokinetic}
\newacro{PD}{pharmacodynamic}
\newacro{MRSA}{methicillin-resistant Staphylococcus aureus}
\newacro{NCPS}{Neural Cyber Physical System}
\newacro{NNV}{Neural Network Verification}
\newacro{ODE}{ordinary differential equation}
\newacro{ReLU}{Rectified Linear Unit}
\newacro{DAE}{differential-algebraic equation}
\newacro{CBF}{Control barrier function}
\newacro{SMT}{satisfiability modulo theories}
\newacro{dL}{differential dynamic logic}
\newacro{TDM}{therapeutic drug monitoring}
\newacro{MIPD}{Model-informed precision dosing}
\newacro{CLI}{command line interface}
\newacro{ITP}{interactive theorem prover}
\newacro{ATP}{automated theorem prover}

\title{Vancomycert: A Certified Neuro-Symbolic Drug Delivery System (Case Study)}
% \title{Neural Network Verification in Healthcare: A Case Study Towards Infinite-Horizon Dosing Safety}

\titlerunning{Vancomycert: Certified Neural Dosing}

% If there are any changes let alistair know (my industrial supervisor wants to
% be kept updated)
\author{
Alistair Sirman\inst{1}\Envelope\orcidlink{0009-0009-3342-089X} \and
Fleur Conway\inst{2, 3}\orcidlink{0009-0005-8990-4202} \and
Jessica Ciupa\inst{2,3}\orcidlink{0009-0003-9505-0608} \and
Gusts Gustavs Grīnbergs\inst{4}\orcidlink{0009-0006-1013-3773}\and
Ekaterina Komendantskaya\inst{1, 3}\orcidlink{0000-0002-3240-0987} \and
Thai Son Hoang\inst{1}\orcidlink{0000-0003-4095-0732} \and
Michael Rawson\inst{1}\orcidlink{0000-0001-7834-1567} \and
Alessandro Bruni\inst{4}\orcidlink{0000-0003-2946-9462}\and
Vaishak Belle\inst{2}\orcidlink{0000-0001-5573-8465}\and
Michael John Williams\inst{5}\orcidlink{0000-0003-4735-7229}
}
\authorrunning{Sirman et al.}

\institute{
  University of Southampton, Southampton, United Kingdom\\
  \email{A.Sirman@soton.ac.uk} \and
University of Edinburgh, Edinburgh, United Kingdom\and
Heriot-Watt University, Edinburgh, United Kingdom\and
IT University of Copenhagen, Copenhagen, Denmark\and
Schlumberger Cambridge Research, Cambridge, United Kingdom
}

\maketitle
\begin{abstract}
  Neural network controllers for autonomous decision-making are well-established
  in cyber-physical systems, yet their deployment in safety-critical healthcare
  settings remains largely unverified. This paper presents a methodology and
  case study for the formal verification of a neural network controller for
  antibiotic dosing, motivated by the challenge of systems that must be
  simultaneously adaptive and provably safe across unbounded time horizons. We
  construct a simplified yet clinically-interpretable %\ac{PK} and \ac{PD}
  model that tracks drug concentration, body temperature, and white blood cell
  count. \textit{Vancomycin} is selected as a representative antibiotic, widely
  prescribed for severe infections yet carrying a narrow therapeutic window,
  where supratherapeutic concentrations risk nephrotoxicity and subtherapeutic
  dosing risks treatment failure. A supervised neural network controller is
  trained on synthetic clinician-style dosing data. We establish formal
  verification of input-output safety properties, specifically verifying a
  property of a neural network that implies an infinite-horizon proof that
  automated dosing never exceeds the supratherapeutic boundary. This system
  property is proven in \rocq using the \vehicle interactive theorem prover
  back-end to integrate the different proof systems. The end result is a
  verification pipeline that allows for a wide variety of treatment approaches
  whilst maintaining safety for each specific patient.
\end{abstract}


% TODO: Focus more on what case study is about and not PK/PD
% TODO: Less abstraction
% TODO: Vehicle generated spec -- could be a focus

\section{Introduction}
\label{section:intro}

Vancomycin is one of the most widely prescribed antibiotics in hospital
settings, serving as a first-line treatment for severe Gram-positive infections
\cite{chun_impact_2021,pereboom_clinical_2019}. Its therapeutic window is
narrow: supratherapeutic concentrations risk nephrotoxicity, while
subtherapeutic dosing risks treatment failure \cite{zhang_revised_2025}. Dosing
is further complicated by substantial inter-patient variability driven by
various factors such as body weight, age, and severity of illness
\cite{marsot_vancomycin_2012}. Even with \ac{TDM} in place, therapeutic targets
are routinely missed: vancomycin-induced nephrotoxicity has been reported in
31.7\% of patients in clinical practice, with elevated concentrations identified
as an independent predictor of mortality \cite{al-maqbali_vancomycin_2022}.
% TODO: Relook at this
% Although \ac{MIPD} tools, tools that use mathematical and statistical models to
% produce individualised doses~\cite{perez-blancoModelInformedPrecisionDosing2022},
% improve on this, they exhibit heterogeneous predictive performance across
% population models \cite{wicha_therapeutic_2021}.\comment[id=MR]{This can be said
%   more clearly}\comment[id=AS]{I have a paper~\cite{chenEthnicRacialDifferences2006}
%   that supports this issue explicitly across all dosing domains that might
%   strengthen this argument}

\ac{MIPD} remains model-dependent, human-mediated, and unable to provide
continuous autonomous control. Therefore, a closed-loop controller that
adaptively responds to patient-specific pharmacokinetic state and evolving
infection markers is an attractive alternative to static dosing protocols.
Neural network controllers are particularly well-suited to this role: they have
been shown to outperform classical rule-based approaches in settings with
complex non-linear dynamics and high inter-patient variability
\cite{poweleitArtificialIntelligenceMachine2023}, making them a natural
candidate for optimising dosing within the narrow vancomycin therapeutic window.
However, the same adaptivity that makes neural network controllers appealing
also introduces safety risks. In the vancomycin setting, where supratherapeutic
concentrations are already a known cause of irreversible renal injury, any
autonomous controller operating without verified safety bounds risks compounding
these outcomes at scale and without human oversight.

In safety-critical domains, the consequences of undetected failure are
irreversible, such as nephrotoxicity in this case study.
Formal verification offers a strong guarantee; that the autonomous controller
\emph{cannot} violate the safety properties as laid out in the specification.
The formal verification of neural network controllers has seen substantial
progress over the past decade. Tools such as Marabou \cite{katz2019marabou},
$\alpha$,$\beta$-CROWN \cite{wang_beta-crown_2021}, and ERAN
\cite{singh_fast_2018} can establish input-output safety properties over all
inputs within a specified domain, and have been demonstrated on controllers in
safety-critical settings including airborne collision avoidance
\cite{katz_reluplex_2017} and autonomous driving \cite{teuber_provably_2024}.

Neural network verifiers consider the network in isolation, ignoring system
dynamics such as pharmacokinetics and the time component. As a result, they are
restricted to properties expressible as linear or polynomial constraints
(depending on the solver) over a single input-output relation, with no capacity
to reason about how doses propagate through \ac{PK}/\ac{PD} dynamics over time.
In addition, closed-loop verifiers such as NNV \cite{tran_nnv_2020} and VeriSig
2.0 \cite{ivanov_verisig_2021} propagate reachable sets through combined
network-plant dynamics, but approximation errors compound over successive steps,
making unbounded-horizon guarantees inaccessible. Also, the proof must hold
parametrically across a wide population of patients whose \ac{PK} parameters
vary with age, weight, and renal function, a form of generalisation that
existing verifiers do not support. An \ac{ITP} such as \rocq addresses these
gaps: allowing a user to do arbitrary non-linear reasoning, establish invariants
over infinite time horizons, and create parametric proofs that range over
patient models, with a type checker to confirm the validity of the proofs, and
can interface with the neural network solvers through the intermediary tools.

We present Vancomycert, a new case study to demonstrate a complete, end-to-end
formally verified neural network controller for a medical dosing task. The key
enabler is \vehicle \cite{FoMLAS2023:Vehicle_Tutorial_Neural_Network}, a
domain-specific language for expressing properties over neural networks that
compiles a single high-level specification to two%\comment[id=AS]{three to two
% if no pdt}
back-ends simultaneously:
%a TensorFlow loss function for property-driven training \comment[id=AS]{We dont do pdt}, 
a set of Marabou queries for automated reasoning upon neural network
verification, and an extraction into \rocq for integration with an \ac{ITP}.
Prior applications of \vehicle have demonstrated either the training and
verification pipeline \cite{FoMLAS2023:Vehicle_Tutorial_Neural_Network} or the
\ac{ITP} integration \cite{daggitt_vehicle_2026} in isolation; Vancomycert is
the first study to exercise the full pipeline end-to-end, using \vehicle as a
single source of truth that
connects %property-driven training\comment[id=AS]{We dont do pdt},
automated verification, and infinite-horizon proof in one unified specification.
% \comment[id=MR]{This paragraph is way too advertisey and doesnt serve much other
%   purpose, and could probably be removed}

This paper makes \chreplaced[id=AS]{two}{three} contributions. First, we present
a realistic, clinically motivated benchmark for neural network verification: a
closed-loop antibiotic dosing controller for vancomycin, with a fully specified
\ac{PK}/\ac{PD} model, a trained neural network controller, and a complete
formal safety specification, providing a novel benchmark for the \ac{NNV}
community. Second, we propose an original methodology for decomposing
infinite-horizon safety verification into two levels of abstraction: a
system-level proof in \rocq that establishes the infinite-horizon, non-linear
invariant over the closed-loop \ac{PK}/\ac{PD} dynamics, and a \ac{NNV}-level
property $\psi$ that is linear and tractable for automated tools, with \rocq
proving that $\psi$ is sufficient to imply the system-level guarantee. This
separation of concerns is the key architectural contribution of the paper.
% Third, we demonstrate that \vehicle can serve as a single source of truth for
% the theorem prover to rely on.
% \comment[id=AS]{This could be a decent point but
%   at the moment it doesnt say much. The advantage is no other tools, no other
%   axioms, so if vehicle is correct, then so is the proof. This also may not be
%   stricyly true, depending on if you include Marabou as part of it}

% \comment[id=AS]{Again, relatively advertisery}
% TODO: Maybe repeated in background, check and if not make less advertisery
% In doing so, we provide a rich
% showcase of \vehicle's expressive capabilities, exercising parametric
% specifications, normalisation, and \ac{ITP} extraction simultaneously within a
% single specification. Therefore, representing an extensive use of \vehicle's
% features, demonstrating the flexibility and breadth of its specification
% language across a realistic, end-to-end pipeline.

\section{Background}
\label{section:background}

\subsection{Pharmacokinetic and Pharmacodynamic Modelling}

\ac{PK} models describe how a drug moves through the 
body over time, capturing the processes of absorption, distribution, 
metabolism, and elimination. \ac{PD} models describe 
the relationship between drug concentration and its biological effect. 
Together, \ac{PK}/\ac{PD} models form the mathematical foundation for 
dosing optimisation, expressing the body's response to a drug as a system 
of \acp{ODE} whose parameters vary between patients according to factors 
such as age, weight, and renal function \cite{marsot_vancomycin_2012}.
The simplest widely-used \ac{PK} model is the one-compartment model, 
which treats the body as a single homogeneous volume and assumes 
first-order elimination: the rate at which a drug is cleared is 
proportional to its current concentration. This model is clinically 
interpretable and verification-tractable, making it the foundation of 
the \ac{PK}/\ac{PD} model used in this case study. While more accurate 
multi-compartment models exist for vancomycin \cite{marsot_vancomycin_2012}, 
the one-compartment formulation preserves the essential dynamics relevant 
to safety verification whilst remaining amenable to formal analysis. \chadded[id=FC]{Moreover, the possible formulations of compartments for both the \ac{PK} and \ac{PD} equations are extensive and a research area in itself~\cite{yoon_model-informed_2023}. Hence, for tractability, the scope is restricted to verification, and simplified equations are used}. \chdeleted[id=AS]{ The}
The key parameters governing this model are the absorption rate 
constant $k_a$, the elimination rate constant $k_e$, and the volume 
of distribution $V_d$, all of which vary between patients according 
to physiological factors such as renal function, body weight, and age 
\cite{marsot_vancomycin_2012}. The full model, including the concentration 
equations and pharmacodynamic infection markers, is presented in 
Section~\ref{section:casestudy}.

\subsection{Neural Network Verification}

Neural network verification (\ac{NNV}) addresses the problem of establishing
that a network's input-output behaviour satisfies a formally specified property
for \textit{all} inputs within a given domain \cite{casadio_neural_2022}. The
most common class of properties are safety and robustness properties. Safety
properties assert that outputs remain within a safe region for all inputs in a
specified domain. Robustness properties assert that small perturbations to an
input cannot change the network's output beyond a specified bound. Current
\ac{NNV} tools, including the Marabou framework \cite{katz2019marabou} used in
this work, operate by encoding the network's computation as a set of constraints
and querying a solver for a counterexample. A key limitation shared by all such
tools is that verifiable properties must be expressible as linear constraints
over the network's input-output relation. Non-linear system dynamics, unbounded
time horizons, and parametric patient variation therefore fall outside their
scope, motivating the \ac{ITP}-based reasoning approach taken in this
case-study.

\subsection{Vehicle}

\vehicle \cite{FoMLAS2023:Vehicle_Tutorial_Neural_Network,daggitt_vehicle_2025} 
is a dependently-typed domain-specific language for writing logical specifications 
for neural networks. Specifications are written over the problem domain rather than 
over normalised network inputs, and compiled automatically to verification back-ends 
and \ac{ITP} extractions. Prior applications of \vehicle have demonstrated individual 
features in relative isolation: hyper-rectangle robustness \cite{casadio_neural_2022},
or robustness properties over vision models \cite{FoMLAS2023:Vehicle_Tutorial_Neural_Network}.
This case study exercises a significantly broader range of \vehicle's specification 
language features simultaneously, and in combination, to handle the demands of a medical
safety-critical domain.

In \vehicle, \textit{Parametric specifications} allow us to write a specification with
constants that are given only at
compile time by being %$k_a$, $k_e$, $V_d$, $C_{safe}$, and
% $t_{td}$ to\comment[id=AS]{Not defined yet, and probably over detailed for this section}
passed as external command-line parameters. This makes the specification generic
across patients without modification, as information about the patient that is
invariant and not known to the neural network can be easily generalised over.
\textit{Dataset declarations} load normalisation statistics from external files,
ensuring the specification operates over the same input distribution as
training. \textit{Network composition} allows
a normalised version of the network to be defined as the composition
of the raw network with the
normalisation transform, so that all properties can be stated directly over
unnormalised clinical inputs. \textit{Conditional branching} in the output
property handles the two approximation directions for the irrational exponential
terms in $\psi$, a property structure that goes beyond simple linear bounds.
Finally, all three \vehicle back-ends
from % property-driven training\comment[id=AS]{We dont do pdt},
Marabou verification, and \rocq extraction are exercised simultaneously from a
single specification, using the \ac{ITP} extraction.

\subsection{Rocq}

\rocq \cite{majumdar_smtcoq_2017} is a proof assistant based on the Calculus of
Inductive Constructions, in which propositions are types and proofs are terms.
It supports machine-checked formal proofs over arbitrary mathematical
structures, including the real analysis required to reason about continuous
\ac{PK}/\ac{PD} dynamics. The formalisation in this paper builds on the MathComp
Analysis library \cite{affeldt_mathcomp-analysis_2026}, which provides the
definitions and theorems needed to reason about continuous functions and
integrals over the reals. The key capabilities of \rocq exploited here are
inductive reasoning over the number of administered doses, and parametric proof
over patient \ac{PK} parameters, establishing that the safety invariant holds
for any instantiation of the \vehicle parameters within the specified clinical
range. The \vehicle-verified property $\psi$ is imported as an axiom via
\vehicle's \rocq back-end and discharged as a hypothesis in the inductive step.

\section{Case Study and Methodology}
\label{section:casestudy}
In order to verify that the patient will not overdose as a result of a
controller, we use the formalisation of the human body presented by the
one-compartment pharmacokinetic
model~\cite{taleviOneCompartmentPharmacokineticModel2021}. This is a set of
\acp{ODE} that approximate the body's reaction to drugs, when the body is
considered to be one continuous volume of matter. Specifically, when an oral
dose $D$ is given, the concentration in the blood increases continuously as the
dose gets absorbed, and then after a period of time starts decreasing. The
following equation introduces a function for concentration at time $t$, modelled
using \emph{exponential decay}:
\begin{equation}\label{eq:single_dose}
  C(D,t)=\frac{D \cdot k_a}{\mathit{Vd} \cdot (k_a - k_e)}\cdot \left(e^{-k_e\cdot t} - e^{-k_a\cdot t}\right)
\end{equation}

\begin{figure}[t]
  \centering
  \resizebox{0.5\textwidth}{!}{
  \input{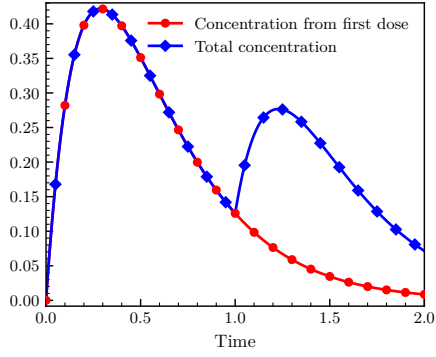}
  }
  \caption{The difference between the concentration from a single dose $C$, and
    the total concentration $\ct$.}
  \label{fig:concentration}
  % alt={A graph showing concentration from one dose and from two. The lines for
  % one dose and two are exactly the same, increasing but following exponential
  % decay, until a second dose is given at time , where the total concentration
  % starts increasing again, while the first dose is still decreasing.}
\end{figure}
\chreplaced[id=AS]{w}{W}here $k_a$ is the absorption constant, $k_e$ is the elimination constant, and
$\textit{Vd}$ is the volume of blood.

This response to a dose can be used to
  model a system with multiple doses at regular intervals, $\textit{ttd}$ (time
to dose), that can calculate an over-approximation of the concentration in the
patients blood at any point in time. $\ct$ is introduced as the total
concentration, and $\bar{D}$ is a list of length $n$ of the doses given, and the
difference of the functions can be seen in \autoref{fig:concentration}.

\begin{equation}
  \ct(\bar{D},t)=\sum_{i=0}^{n}(\text{max}(0, C(\bar{D}_i, t - ttd \cdot i)))
  \label{eq:sum_doses}
\end{equation}

Note that this is an over-approximation because it is the sum of the individual
doses concentration, and thus does not take into account the fact that the
body eliminates more aggressively for higher total concentration.
Also, \autoref{eq:sum_doses} is continuous and differentiable everywhere except
where $t$ is a multiple of $ttd$, where it is only continuous.

We must also define what information is required for accurate dosing, and what a
clinically viable patient is. The information used to make an accurate dosing
decision is the patient's current concentration $c$, temperature $T$, white
blood cell count $\textit{WBC}$, age and weight. We consider a patient is
clinically viable if they are between the ages of 18 and 89, weigh between
\qtylist{50;100}{kg}, have a $\textit{WBC}$ between
\qtylist{7.5;20}{\unit{\kilo\per\micro\litre}}, a temperature between
\qtylist{36.5;40}{\celsius} and the drug concentration in their blood is between
0 \unit{\milli\mole\per\litre} and $\csafe$ \unit{\milli\mole\per\litre} at the
start of treatment.

With this model, we can specify our safety property formally:
% \begin{theorem}[Main Safety Theorem]\label{thm:safety}
\begin{definition}[Main Safety Definition]\label{thm:safety}
  % For any clinically viable patient, for any valid sequence of doses, the
  % concentration of Vancomycin in the blood does not exceed the supratherapeutic
  % boundary, $\csafe$, at any time.
  For any clinically viable patient, a valid sequence of doses is a list of real
  numbers such that the total concentration of Vancomycin in the blood when
  given regular doses. where the size of the dose is determined by the list, does not
  exceed the supratheraputic boundary, $\csafe$, at any time.
  \begin{equation*}
    \exists \bar{D} : list\ \mathbb{R}, \forall t : \mathbb{R}, \ct(\bar{D}, t)\leq \csafe
  \end{equation*}
\end{definition}
% \end{theorem}
This poses the question:
\begin{center}\itshape
  How does one formalise this property?
\end{center}

To answer this, we first introduce an arbitrary function $f \colon \mathbb{R} \to \mathbb{R}$ that is
responsible for choosing a single dose based on the current concentration. This
function is applied at regular intervals, $\textit{ttd}$, and all doses given are
recorded in a list.
% is then used in a recursive function to calculate the list $\bar{D}$, of
% length $n + 1$, of doses chosen. $\textit{ttd}$ is also introduced as the time between
% doses, and is restricted to be greater than or equal to the time of peak
% concentration in $C$. This is so that whenever a new dose is given, $C$ must be
% non-increasing.
\begin{equation}\label{eq:doses}
  \begin{aligned}
    \textit{doses} \colon \mathbb{R} \rightarrow \mathbb{N} \to \textit{list } \mathbb{R}\\
    \textit{doses}(c, n)=\begin{cases}
      [f(c)] & n = 0\\
      \textit{doses}(c, n-1) \concat [\dblcolon f(\ct(\textit{doses}(c, n-1), \textit{ttd}\cdot n))] & n > 0
    \end{cases}
  \end{aligned}
\end{equation}
This definition allows for an updated theorem where $\bar{D}$ is instantiated with
$\mathit{doses}$.
% , and \emph{clinically viable} is defined as any patients who is not
% currently over-dosing $c\leq\csafe$.
\begin{theorem}[Main Safety Theorem]\label{thm:psi_safety}
  For any initial concentration $c$ less than or equal to $\csafe$, for an
  arbitrary amount of doses, and for any dosing interval $\textit{ttd}$ such
  that $\textit{ttd}$ is at least as long as the time for the concentration from
  the previous dose to peak, the \chadded[id=AS]{total} concentration at any
  point in time $t$ must be less than or equal to $\csafe$.
  \begin{equation*}
    \forall c : \mathbb{R}, \forall t : \mathbb{R}, \forall n : \mathbb{N}, c \leq \csafe \to\ct(\mathit{doses}(n, c), t) \leq \csafe
  \end{equation*}
\end{theorem}

This introduces our next question:
\begin{center}\itshape
  what property $\psi$ of $f$ is sufficient to
prove \autoref{thm:psi_safety}?
\end{center}

To derive $\psi$, we consider the point in time where the decision is being made by
$f$. The current concentration is known to be $c$, and $c$ is less than
$\csafe$. Thus, as long as the network proposes a dose such that concentration
will not exceed $\csafe - c$, it is a safe dose. Furthermore, this
constitutes an invariant across an infinite time horizon. This can be expressed
as follows:
\begin{equation}\label{eq:psi}
  \forall c : \mathbb{R}, \forall t : \mathbb{R}, c\leq\csafe\to c + C(f(c), t)\leq\csafe
\end{equation}
Although a property of $f$, this is can be expanded to a property linear with
respect to $f$, a requirement for a later step. To do so, we
note that $t$ where $\frac{dC}{dt}=0$ is the peak concentration, and
therefore $C(f(c), t) \leq C\left(f(c), \frac{\text{ln}\left(k_a/ k_e\right)}{k_a-k_e}\right)$.
With this, $C$ can both be fully instantiated and unfolded to get a $\psi$ that is
sufficient:
% we note that $C$'s second derivative is strictly
% \replaced[id=AS]{non-positive}{decreasing} with respect to time, and thus
% $\frac{dC}{dt}=0$ is the peak concentration, so our property $\psi$ reveals
% itself.\comment[id=AS]{This is not at all obvious. Firstly, $t$ should be $t$
%   when $\frac{dC}{dt}=0$, and this should be explained}
% \comment[id=AS]{This equation is simply wrong :(}
\begin{equation}\label{eq:f_prop}
  \psi(f) = \forall c : \mathbb{R}, c\leq\csafe\to c+ \frac{f(c) \cdot k_a}{\mathit{Vd} \cdot (k_a - k_e)}\cdot \left(e^{-k_e\cdot \frac{\text{ln}\left(\frac{k_a}{k_e}\right)}{k_a-k_e}} - e^{-k_a\cdot \frac{\text{ln}\left(\frac{k_a}{k_e}\right)}{k_a-k_e}}\right)
  \leq \csafe
\end{equation}
\textit{Note that $k_a\neq k_e$, otherwise there is division by 0.}

This equation for $\psi$ remains sufficient, as the steps taken to arrive at it
were all simply transitivity of less than or equal to. This form is fully
expanded and is linear with respect to $f$ -- that is, neither any
of the inputs or outputs of the network interact with each other in a non-linear
way. This leads to the penultimate question:
\begin{center}\itshape
How does one verify this when
$f$ is a neural network?
\end{center}

The first thing to note is the dimensionality of $f$. It is simply a function
from $\mathbb{R}\to\mathbb{R}$. As described above, this is not what we expected
from our network, as it is expected to know about age, weight\chadded[id=AS]{,} etc. also. This is
because the \chreplaced[id=AS]{theorem}{verification} only relies on the
concentration, $f$ can be a \emph{partially applied} network. \chreplaced[id=AS]{The
  partially applied network is enough to reason about the theorem with, and will
  simplify any proofs}{This is because,
concentration, the target of our verification, has no dependence on the other
inputs and thus can be ignored}.\chadded[id=AS]{ The neural network still cares
  about these inputs, because they are used to determine how unwell a patient is,
  and therefore how much to treat them.} In the case \chreplaced[id=AS]{where
  the theorem relies on more inputs}{it did}, these would also need to be
reasoned upon\chdeleted[id=AS]{, but the parameters that do affect
  concentration are baked into the constants $k_e, k_a, \textit{Vd}$}. In our
example\chadded[id=AS]{, the partially applied network can be formulated as
  such}:
\begin{align*}
  \network : \mathbb{R}^5\to\mathbb{R}\\
  f(c) = \network(c, T, \textit{WBC}, \mathit{age}, \mathit{weight})
\end{align*}
Where $T, \mathit{WBC}, \mathit{age}, \mathit{weight}$ are all fixed based upon
the initial patient state. Although $t$ and $\textit{WBC}$ do actually vary with
time, the property $\psi$ is only with respect to concentration and avoiding a
supratherapeutic concentration, for which these are irrelevant, so the variation
does not impede the safety of the system.

Now we can express $f$ in terms of $\network$, \vehicle is used to check
\autoref{eq:f_prop}. \vehicle requires all properties that are going to be
verified to be linear, hence the earlier effort spent ensuring that $\psi$ can be
expressed as such. This tool is used to specify and verify \autoref{eq:f_prop}
against networks to check if they meet the safety requirements.
However, $e^{-k_e\cdot \hdots}$ and $e^{-k_a\cdot \hdots}$ are potentially
irrational numbers and thus are neither computable or expressible in \vehicle.
To circumvent this issue, we introduce over and under approximations for
$e^{-k_a\cdot \hdots}$ and $e^{-k_e\cdot \hdots}$, and substitute these in place of the
potentially irrational components, which are calculated externally and passed
into \vehicle through the \ac{CLI}. The reason that so many approximations are
required is due to how $k_e$ and $k_a$ are ordered. If $k_e<k_a$:
\begin{align*}
  \frac{D\cdot k_a}{\textit{Vd}\cdot(k_a-k_e)}<0 &&e^{-k_e\cdot \hdots}-e^{-k_a\cdot \hdots}<0
\end{align*}
else $k_a<k_e$:
\begin{align*}
  \frac{D\cdot k_a}{\textit{Vd}\cdot(k_a-k_e)}>0 &&e^{-k_e\cdot \hdots}-e^{-k_a\cdot \hdots}>0
\end{align*}
Therefore, if $k_e<k_a$, then to over approximate $C$, one must under approximate
$e^{-k_e\cdot\hdots}-e^{-k_a\cdot\hdots}$, and thus under approximate $e^{-k_e\cdot\hdots}$ and
over approximate $e^{-k_a\cdot\hdots}$. However, if $k_a<k_e$, then the direction of
all approximates must flip.

% Because $k_a$ may be
% smaller or greater than $k_e$, it is unknown weather or not $e^{-k_e\cdot t}$ or
% $e^{-k_a\cdot t}$ should be under or over approximated, as the polarity of the
% fraction and bracket in~\autoref{eq:f_prop} changes, and thus the direction of
% the approximation changes with it.

With this, the final question is posed:
\begin{center}\itshape
  How does one train a neural network that meets the specification?
\end{center}

Training a neural network must be done in tandem with the development of a
specification for verification. First, data must be generated to provide input
and output to the network, which exposes the neural network to boundary
examples, i.e. the threshold values within the requirements, such as maximum
concentration limits. Second, a technique called property-driven training~\cite{flinkow2025generalframeworkpropertydrivenmachine}, can
be implemented - the loss function is adapted to contain the specification,
hardening the neural network against the constraints. However, this is not
covered in this paper.

To summarise the methodology, it is equivalent to answering these questions:
\begin{enumerate}\itshape
\item How do we formalise the system?
\item What is our main safety property?
\item What property of an opaque function would implicate the main safety property?
\item How can we train a neural network to meet that property?
\end{enumerate}

There was also an implicit decision about how to split what belongs to \vehicle
and what to \rocq. As the system is inherently non-linear, the \vehicle
component must somehow become linear. In this case study, we found that the
property $\psi$ is in fact linear, and thus all the non-linear system
modelling can be off-loaded to \rocq.

\section{Network and Training}
Motivated by antibiotic dosing, a neural network controller which could act as a
test-bed for the infinite-horizon properties, was generated. A feed forward
multi-layer perceptron (MLP) was trained on synthetic patient state and dosing
data\chadded[id=AS]{. }% \deleted[id=AS]{, with property-driven training to enforce these constraints within the loss
%function}.
In this work, the patient is modelled using a low-dimensional state
consisting of drug concentration (C), body temperature (T), white blood cell
count (WBC), and static demographic variables (age, weight). The controller
outputs a discrete antibiotic dose administered at 12 hour intervals, reflective
of typical vancomycin dosing regimes. Between dosing events, the patient state
evolves continuously according to pharmacokinetic and pharmacodynamic equations.
The aim is to dose when the patient is unhealthy, without reaching unsafe
states, and not to dose when the patient is healthy.

\subsection{Data Generation}\label{sec:data_gen}

To train the neural network controller, we first constructed a synthetic,
clinician-style dataset using a custom pharmacokinetic/pharmacodynamic (PK/PD)
simulation. Synthetic train and test data were generated by forward simulation
of the PK/PD model under a rule-based feedback dosing controller over a 24 day
period. Euler updates for each state variable, alongside intermittent dose rate
adjustments, occurring at intervals of 12 hours. Data is generated for a cohort
of 50 patients, featuring randomly sampled ages between 18 and 90 years, weighs
between 50.0 and 100.0 kg, and binary biological sex assignments. State-action
pairs—comprising concentration, temperature, white blood cell count, age,
weight, and the resulting target dose—are recorded at every 12 hour time-step.
\chdeleted[id=FC]{The resulting dataset is randomly partitioned into an 80/20 training and testing
split.} \chadded[id=FC]{While synethic data was used,} \chdeleted[id=FC]{The split was used with the synthetic data so that}real data could be
easily substituted. To create more realistic patient data, coefficients of the
PK/PD equations are impacted by age, weight etc.

\chadded[id=FC]{Body temperature $T$ follows first-order dynamics capturing homeostasis, infection-driven increase, and drug-mediated reduction}:

\begin{equation}\label{eq:temp_update}\frac{dT}{dt} = r_T(T) - k_T^C C - k_T^{hom}(T - T_{norm})\end{equation}

% \comment[id=MR]{You tell me that you've done a load of fancy modelling to generate the data that trains the network. Leaving the aside the obvious circularity in using a *model* to train a *model* that you then verify against a *model*, I'm overall not sure why this matters for the case study.}
% \comment[id=AS]{NN is probably easier to run than rthe model, so for
%   exportability this is better. also verification is still generic so its
%   somewhat irrelevant}
where $r_T(T)$ is the infection-driven temperature increase, $k_T^C$ is the
drug temperature coupling coefficient, $k_T^{hom}$ is the homeostatic return rate,
and $T_{norm}$ is the normal body temperature.
% \chcomment[id=AS]{Missing introductory text here}

\chadded[id=FC]{Similarly, the change of white blood cell count $WBC$ with respect to time}:
\begin{equation}\label{eq:wbc_update}\frac{dWBC}{dt} = r_{WBC}(WBC) - k_{WBC}^C C - k_{WBC}^{\text{hom}}(WBC - WBC_{\text{norm}})
\end{equation}

where $r_{WBC}(WBC)$ is the infection-driven WBC production term, $k_{WBC}^C$ is
the drug WBC coupling coefficient, $k_{WBC}^{\text{hom}}$ is the homeostatic decay
rate, and $WBC_{\text{norm}}$ is the healthy baseline WBC count.

At each time-step, the simulation calculates a heuristic clinical dose $D_t$
based on the patient's vitals and current drug concentration. To maintain
physiological plausibility, the updated temperature and white blood cell counts
are strictly constrained within predefined clinical limits using hard state
clipping.

The pharmacokinetic (PK) model assumed oral drug administration with first-order absorption and elimination. Following each administered dose $D_i$, the resulting drug concentration is governed by a standard one-compartment model as described in Equation \ref{eq:single_dose}. As in Equation \ref{eq:sum_doses} \footnote{$V_{d, \mathrm{eff}}$ is computed here, distinct from
the baseline population parameter $\textit{Vd}$, as it is individually
adjusted based on the simulated patient's age and weight covariates. While
fine for data generation and training, it cannot be used in $\psi$ as it would
make the property non-linear.}, the total concentration at time $t$ is computed as the sum of contributions from all previously administered doses.

\subsection{Training}

The dosing controller is implemented as an MLP \cite{GARDNER19982627} using the
Keras Sequential API. The network architecture consists of an input layer for
the 5-dimensional state vector, followed by two hidden layers containing 128 and
64 neurons, respectively. \chadded[id=AS, comment={1, 3, maybe 4. Maybe move to
  ``limitations``}]{A small network has been chosen as it is not meant to be
  performant, instead it is to show that a network can be trained that holds the
  property $\psi$. As the methodology relies \chadded[id=FC]{on} neural network solvers, the
  limitations on size/complexity of operators is limited only by the solver
  chosen. In this case, Marabou is used, which can handle networks that are as
  large as $n$ nodes, but is limited to only ReLU operators.} Both hidden layers
utilise Rectified Linear Unit (ReLU)\cite{DBLP:journals/corr/abs-1803-08375}
activation functions to capture the non-linear relationships between
physiological state variables and dosing decisions. The final output layer
consists of a single neuron that also applies a ReLU activation, to ensure
non-negative dose predictions. A small positive offset is introduced during the
model export to guarantee this strict positivity for formal verification. During
the baseline supervised phase, the model is trained using the Adam optimiser
\cite{kingma2017adammethodstochasticoptimization} to minimise Mean Squared Error
(MSE). The optimisation is performed over 50 epochs with a batch size of 32. The
dataset is split into training and test sets using an 80/20 split.

\section{Verifying Neural Networks with \vehicle}\label{sec:vehicle}
In order to verify the property $\psi$ specified in~\autoref{eq:f_prop}, we use
\vehicle. This allows the checking of any expressible linear assertion. We
initially specify the types and assign meaningful names to the indices of the
input domain\footnote{The ``;`` syntax to simulate new lines is not syntactically valid in
  \vehicle, it is only used here for brevity.}.
\begin{minted}[fontsize=\listingsize]{vehicle}
type UnnormalisedInputVector = Tensor Real [5]
type InputVector = Tensor Real [5]

conc = 0; temp = 1; wbc = 2; age = 3; weight = 4

type OutputVector = Tensor Real [1]
\end{minted}
After this, normalisation is defined, as the network was normalised in training,
the inputs described later are meaningful and then normalised. The normalisation
vectors are created from the data and saved to files that are passed to
\vehicle's \ac{CLI}, and are expressed explicitly as a transform in \vehicle.
\begin{minted}[fontsize=\listingsize]{Vehicle}
@dataset
meanScalingValues : Tensor Real [5]

@dataset
standardDeviationValues : Tensor Real [5]

normalise : UnnormalisedInputVector -> InputVector
normalise x = foreach i .
  (x ! i - meanScalingValues ! i) / (standardDeviationValues ! i)
\end{minted}
The neural network, \texttt{pk}, is expressed as mapping from the
\texttt{InputVector} to the \texttt{OutputVector}. The properties of interest for verification,
however, relate to the mapping, \texttt{normpk}, between the \texttt{UnnormalisedInputVector}
and the \texttt{OutputVector}. This allows verification over the patient specific
information that reflect realistic values, while still maintaining normalisation.
It also takes all the
parameters needed to construct a \ac{PK} model to verify against. This is how it
manages to be generic over each patient, as this information is patient
specific\footnote{The ``,`` syntax is not valid, but used here for brevity. Refer
  to~\aref{appendix:lst:vcl_spec} for the correct version.}.
\begin{minted}[fontsize=\listingsize]{Vehicle}
@network
pk : InputVector -> OutputVector

normpk : UnnormalisedInputVector -> OutputVector
normpk x = pk (normalise x)

@parameter
Ka, Ke, Vd, C_safe, ttd, Ka_over, Ka_under, Ke_over, Ka_under : Real
\end{minted}
Compile-time constant folding produces a linear specification by instantiating
the parameters with the given values. These
parameters are imperative for the proof, as the specification should be
generic and the only source of truth relied upon by the proof, they must be
encoded here. For example the fact that $k_a$ must not equal $k_e$:
\begin{minted}[fontsize=\listingsize]{Vehicle}
@property
Ke_n_Ka : Bool
Ke_n_Ka = Ka != Ke
\end{minted}
As these are known at compile time, it is an easy check that does not require an
SMT solver to verify this or similar properties. All of these properties about
parameters can be found along in the full specification file in
\aref{appendix:lst:vcl_spec}.

The approximations \texttt{Ke\_under}, \texttt{Ke\_over}, \texttt{Ka\_under} and \texttt{Ka\_over} are
introduced similarly but are \emph{not} checked by \vehicle. Instead, it is left to the chosen
\ac{ITP} to verify:
\begin{align*}
  \texttt{Ke\_under}\leq e^{-k_e\cdot\frac{\text{ln}\left(\frac{k_a}{k_e}\right)}{k_a-k_e}}\leq\texttt{Ke\_over}\\
  \texttt{Ka\_under}\leq e^{-k_a\cdot\frac{\text{ln}\left(\frac{k_a}{k_e}\right)}{k_a-k_e}}\leq\texttt{Ka\_over}
\end{align*}

\autoref{thm:safety} described what counts as a clinically viable patient, in
\vehicle a corresponding hyper-rectangle input sub-domain can be defined.
\begin{minted}[fontsize=\listingsize]{Vehicle}
safeInput : InputVector -> Bool
safeInput x =
    0 <= x ! conc <= C_safe and
    36.5 <= x ! temp <= 40 and
    7.5 <= x ! wbc <= 20 and
    18 <= x ! age <= 89 and
    50 <= x ! weight <= 100
\end{minted}
Finally, the output domain restriction can be specified based
on~\autoref{eq:f_prop}, along side the the property \vehicle
verifies.
\begin{minted}[fontsize=\listingsize]{Vehicle}
safeOutput : InputVector -> Bool
safeOutput x = let y = ((((normpk x) ! 0) * Ka) / (Vd * (Ka - Ke))) in
           if Ka < Ke
           then (x ! conc) + y * (Ke_under - Ka_over) <= C_safe
           else (x ! conc) + y * (Ke_over - Ka_under) <= C_safe

@property
safe : Bool
safe = forall x . safeInput x => safeOutput x
\end{minted}
This property is sufficient for~\autoref{eq:f_prop} and can not only
be tested by \vehicle, but can also be extracted to an Interactive Theorem
Prover for system verification.

\vehicle also verifies a non-negativity property. This is obviously true as it
is encoded in the neural network that was trained, however, as this
specification should be generic over \emph{any} network, and \rocq needs knowledge
of this fact, it is included in the specification.
\begin{minted}[fontsize=\listingsize]{Vehicle}
@property
nonNeg : Bool
nonNeg = forall x . safeInput x => 0 <= (normpk x) ! 0
\end{minted}

\section{Verifying System Dynamics with \rocq}
In order to show that~\autoref{eq:f_prop} is sufficient to
prove~\autoref{thm:psi_safety}, we formalise the theorem in \rocq and use the
\vehicle extraction to have a consistent proof system. To formalise the system,
we implement~\autoref{eq:single_dose} and \autoref{eq:sum_doses}, respectively.

\begin{minted}[fontsize=\small]{coq}
Definition Concentration (D t : R) : R :=
  ((D * Ka) / (Vd * (Ka - Ke))) * (expR ((-Ke) * t) - expR ((-Ka) * t)).
\end{minted}
\begin{minted}[fontsize=\small]{coq}
Definition total_conc {n} (Ds : n.-tuple R) : R -> R :=
  \sum_(i < n) ((cst 0) \max (Concentration (tnth Ds i))
            \o (center (ttd * i%:R))).
\end{minted}
(Here, cst is the constant function, \texttt{tnth Ds i} returns the $i$th element of
Ds, and center$(x, y)=y-x$).
% (Here, \texttt{cst} is $\lambda x \rightarrow \lambda y \rightarrow x$, \texttt{f \textbackslash max g} is
% $\lambda x\rightarrow \text{max}(f(x), g(x)), $\texttt{tnth Ds i} is the $i$th
% element of \texttt{Ds}, \texttt{center} is
% $\lambda x \rightarrow \lambda y \rightarrow y - x$ and \texttt{\textbackslash
%   o} is function composition).

With these definitions, we can introduce an opaque representation of the reduced
network, along with the properties and definition that are sufficient to prove the main
theorem.
\begin{minted}[fontsize=\listingsize]{coq}
Context {network : R -> R}.

Hypothesis safe : forall C : R, 0 <= C <= C_safe ->
  C + (Concentration (network C) dCdt_root) <= C_safe.

Hypothesis non_neg : forall C : R, 0 <= C <= C_safe -> 0 <= network C.

Fixpoint n_doses (initial : R) (n : nat) : n.+1.-tuple R :=
  match n with
  | 0 => [:: network initial]
  | n'.+1 =>
      let Doses := n_doses initial n' in
      rcons Doses (network (total_conc Doses (ttd *+ (n'.+1))%R))
  end.
\end{minted}
Note that \texttt{safe} and \texttt{non\_neg} are exactly the properties laid out in the
\vehicle specification file, except they act upon the reduced \texttt{network}.

Finally, \autoref{thm:psi_safety} can be expressed in \rocq.
\begin{minted}[fontsize=\listingsize]{coq}
Theorem doses_safe (n : nat) (initial t : R) (HC : 0 <= initial <= C_safe) :
  0 <= total_conc (n_doses initial n) t <= C_safe.
\end{minted}

This theorem is proven in \rocq, and can be found in
\cite{Sirman_Vancomycert_Rocq_Files_2026}.
A proof using the same notation as~\autoref{thm:psi_safety} is provided
below:
\begin{proof}
  Because the $\bar{D}$ is derived from $doses$, its length is $n+1$. We do induction
  on $n$ on the case where $\bar{D}$ has only 1 element, and when it
  has multiple.

  In the case where $n=0$, and $\bar{D}=[f(c)]$:
  \begin{align*}
    c \leq \csafe \rightarrow \ct(doses(0, c), t)\leq \csafe & & \text{Substitution}\\
    c\leq \csafe\rightarrow \ct([f(c)], t)\leq \csafe & & \text{Definition of }doses\\
    c\leq \csafe\rightarrow \text{max}(0, C(f(c), t))\leq \csafe & & \text{Definition of }\ct
  \end{align*}
  We consider the case
  $C(f(c),t)<0$. Clearly, if $C(f(c),t)<0$, then
  $\text{max}(0,C(f(c),t))=0\leq \csafe$.
  In the other case $C(f(c), t)\geq 0$:
  \begin{equation*}
    c\leq\csafe\rightarrow C(f(c), t)\leq\csafe
  \end{equation*}
  As $c$ is non-negative, this can be weakened:
  \begin{equation*}
    c\leq\csafe\rightarrow c + C(f(c), t)\leq\csafe
  \end{equation*}
  We have derived exactly~\autoref{eq:psi}!

  For the inductive case $n=m+1$, $\bar{D}=d::x::xs$, we proceed by cases on $t$
  and $\textit{ttd}\cdot n$. We also remove the assumption $c\leq\csafe$
  for brevity, as it is not required for this portion of the proof. If
  $t\leq\textit{ttd}\cdot n$, then the newest dose has not been given yet or it has only
  been given in that instant and had no effect. It
  follows directly from the inductive hypothesis. Exactly:
  \begin{align*}
    \ct(d::x::xs, t)\leq\csafe & & \\
    \ct(x::xs, t) + \text{max}(0, C(d, t-\textit{ttd}\cdot n))\leq\csafe & & \text{Definition of }\ct\\
    \ct(x::xs, t) + 0 \leq\csafe & &\text{Resolution of max}
  \end{align*}
  In the case $t < \textit{ttd}\cdot n$, the dose has had time to have some effect.
  \begin{align*}
    \ct(d::x::xs, t)\leq\csafe & & \\
    \ct(x::xs, t) + \text{max}(0, C(d, t-\textit{ttd}\cdot n))\leq\csafe & & \text{Definition of }\ct\\
    \ct(x::xs, t) + C(d, t-\textit{ttd}\cdot n)\leq\csafe & & \text{Resolution of max}
  \end{align*}
  Because $\ct(x::xs, t)$ was given at least $\textit{ttd}$ units of time ago, and
  because $\textit{ttd}$ is greater than or equal to the root of the first derivative, it must
  be non-increasing, thus $\ct(x::xs, t)\leq\ct(x::xs, \textit{ttd}\cdot n)$:
  \begin{equation*}
    \ct(x::xs, \textit{ttd}\cdot n) + C(d, t-\textit{ttd}\cdot n)\leq\csafe
  \end{equation*}
  We know based on the definition in~\autoref{thm:psi_safety}, $\bar{D}$ is produced by
  $\textit{doses}$, so $d$ being the head of $\bar{D}$ implies:
  \begin{align*}
    d=f(\ct(x::xs, \textit{ttd}\cdot n))\\
    \ct(x::xs, \textit{ttd}\cdot n)+C(f(\ct(x::xs, \textit{ttd}\cdot n)), t-\textit{ttd}\cdot n)\leq\csafe
  \end{align*}
  For simplicity, we let $c'$ be the concentration at the time the $n$th dose is
  given, $c'=\ct(x::xs, \textit{ttd}\cdot n)$.
  \begin{equation*}
    c'+C(f(c'), t-\textit{ttd}\cdot n)\leq\csafe
  \end{equation*}
  Again, this fits perfectly with~\autoref{eq:psi}, where $t=t-\textit{ttd}\cdot n$ and $c=c'$!
  Therefore,~\autoref{eq:psi} is sufficient to prove~\autoref{thm:psi_safety}.
\end{proof}

\section{Related Work}
\label{section:relatedworks}

\subsubsection{Neural Network Verification Tools and Benchmarks}
\ac{NNV} has seen significant growth over the past decade. Reluplex
\cite{katz_reluplex_2017}, an extension of the Simplex method adapted to handle
\ac{ReLU} activations, was among the first tools to provide sound and complete
verification for deep networks. The field has since seen rapid progress in bound
propagation methods, with $\beta$-CROWN \cite{wang_beta-crown_2021} establishing
state-of-the-art performance on robustness benchmarks by combining per-neuron
split constraints with efficient bound propagation, consistently placing first
in the international VNN-COMP competition \cite{kaulen_6th_2025}, which provides
standardised benchmarks and evaluation infrastructure for the community.
Open-loop tools verify network properties in isolation and cannot establish that
a closed-loop control system remains safe over time. Closed-loop verifiers such
as NNV \cite{tran_nnv_2020} and VeriSig 2.0 \cite{ivanov_verisig_2021} address
this by propagating reachable sets through combined network-plant dynamics, but
are fundamentally limited to finite time horizons: approximation errors compound
over successive steps, making unbounded-horizon guarantees inaccessible.

\subsubsection{Infinite-Horizon and Closed-Loop Safety}

Beyond finite-horizon reachability, a complementary line of 
work has pursued forward invariance through barrier certificates. 
\acp{CBF}, introduced for \ac{ODE} systems by 
Ames et al.~\cite{ames_control_2019}, provide computationally efficient 
online safety filters by reducing the invariance condition to a 
quadratic program solved at each time-step. Dawson et al.~\cite{dawson_safe_2023} 
provide a comprehensive survey of neural barrier certificates 
and contraction metrics, covering the fundamental tension between 
the expressiveness of learned certificates and the difficulty of 
formally verifying their correctness. Recent work has extended 
this framework with tools such as FOSSIL \cite{abate_fossil_2021}, 
enabling formal safety guarantees for learned controllers. 
Zhang et al.~\cite{zhang_verification_2026} extend this further to \ac{DAE} 
systems, providing \ac{SMT}-based verification for 
polynomial and neural network barrier candidates, targeting 
cyber-physical systems such as power networks and robotic 
manipulators.

A particularly relevant line of work draws on \ac{dL}
\cite{DBLP:journals/jar/Platzer17}, implemented in the KeYmaera X proof assistant,
which supports infinite-time safety invariants for control envelopes where
governing differential equations admit no closed-form solution. Teuber et
al.~\cite{teuber_provably_2024} build on this with VerSAILLE, which bridges
\ac{dL} contracts and open-loop \ac{NNV}: given a control envelope proven safe
in KeYmaera X, they derive a specification for the neural network controller and
show that any network satisfying this specification inherits the infinite-time
safety guarantee, demonstrated on airborne collision avoidance benchmarks
\cite{katz_reluplex_2017,johnson_arch-comp21_2021}. This represents the closest
prior work in spirit to \textsc{Vancomycert}: both decompose the verification
problem into a tractable \ac{NNV} property and a separate proof in a formal
system. Where the approaches diverge is that VerSAILLE derives specifications
from an existing \ac{dL} contract, whereas \textsc{Vancomycert} constructs its
invariant proof from first principles over a \ac{PK}/\ac{PD} \ac{ODE} system,
with \rocq proving that the linearisation is sound and handling all non-linear
reasoning within the proof assistant.

\subsubsection{ITP Integration and Proof Certificates}
\label{sec:certs}
A central challenge in deploying formally verified neural network 
controllers is establishing a trusted chain of evidence from the 
automated verifier to the proof assistant. As noted in 
Section~\ref{section:background}, \vehicle currently produces opaque 
axioms rather than checkable proof certificates. The Sledgehammer 
framework for Isabelle demonstrates one resolution: external 
\ac{ATP} and \ac{SMT} solvers discharge proof obligations whose 
results are reconstructed inside Isabelle's trusted kernel, with 
recent work extending this to the cvc5 solver via the Alethe 
certificate format \cite{majumdar_smtcoq_2017}. The analogous 
infrastructure for \rocq is provided by SMTCoq 
\cite{majumdar_smtcoq_2017}, which replays \ac{SMT} solver 
results inside Coq with a verified checker rather than accepting 
them as axioms, representing the direction future \vehicle 
back-ends could take to close the certificate gap identified in 
this paper.

\section{Discussion}
This case study is applicable to \emph{any} oral pharmaceutical, as a different drug
affects only the constants, not the general shape of the function. The fact that
is generic over patients means a \ac{MIPD} \cite{poweleitArtificialIntelligenceMachine2023} (a system for dosing
patients based on their specific circumstance, instead of for the general
population) could be integrated to potentially see better results.

This also showcases a way to circumvent current limitations of neural network
solvers, specifically the linearity constraint and lack of infinite-horizon
guarantees. By integrating with an \ac{ITP}, it allows us to offload the non-linear
and infinite-horizon reasoning to a powerful theorem prover that does not suffer
from such limitations, whilst still being able to verify properties of the
networks that are required.

This specific example was made possible by finding a linear invariant that
implies our over-arching property. As provers improve, and as more support
VNNLIB, a standard query format for neural network
solvers~\cite{roy2026vnnlib20rigorousfoundations}, less restrictive provers will be
able to interface with Vehicle, and the limitations on what can be proven will
be relaxed.
\subsection{Limitations}
As mentioned in Section \ref{sec:certs}, \vehicle does not produce a full proof
with a certificate, leading to potential issues in the extraction mechanism to
propagate and invalidate the proof. Furthermore, as SMT-solvers are generally
not proven to be sound, a bug in the SMT-solver could also invalidate the proof.
The limitations of neural network solvers still reduces what can be feasibly
verified, as not all systems have a linear property of the network that is
sufficient.

As with any use of both machine learning and \ac{ITP}s, a significant time
investment is required to train the model, derive properties like $\psi$ and
formally prove them, restricting the appeal for industrial applications where
there is no safety-critical property that must be verified. With this in mind,
the neural network that was trained has not been evaluated, as its purpose was
to show that \vehicle could be used to train a network, alongside traditional
data-driven training, that meets the specification, and not to produce an
optimal network. Training bigger networks is possible, with Marabou currently
able to handle up to 68 million nodes~\cite{brix2024fifthinternationalverificationneural}.
\subsection{Future Directions}
We believe that the methodology laid out here for finding and proving the
invariant is generic over a much bigger class of problems. In fact, we believe
that any set of differential equations that are continuous and differentiable
almost everywhere across an infinite-horizon, and has a global peak, should be
solvable in a similar manner. It is our hope to generalise the \rocq portion
over this class so that further work on these problems requires minimal work on
the reasoning side, and is applicable to a host of both research and industrial
areas. For instance, in this case, the neural network is trained on a
pre-established dosing policy. However, a more relevant use case would be to
verify an offline RL trained neural network for optimal patient outcomes via
antibiotic dosing. This could be extended to other regimes, such as vasopressor
management \cite{komorowski2018}.

We also hope to bridge the gap between the proof systems of \vehicle and \ac{ITP}s,
by formalising a proof-certificate language in an \ac{ITP}, and allowing \vehicle to
give full proof certificates instead of axioms.

% \input{sections/maths.tex}
% \input{sections/results.tex}

% TODO: messed up citation
\appendix
\section{Vehicle Specification}
\label{appendix:lst:vcl_spec}
\begin{minted}[fontsize=\listingsize]{vehicle}
type UnnormalisedInputVector = Tensor Real [5]
type InputVector = Tensor Real [5]

conc = 0; temp = 1; wbc = 2; age = 3; weight = 4

type OutputVector = Tensor Real [1]

@dataset
meanScalingValues : Tensor Real [5]

@dataset
standardDeviationValues : Tensor Real [5]

normalise : UnnormalisedInputVector -> InputVector
normalise x = foreach i .
  (x ! i - meanScalingValues ! i) / (standardDeviationValues ! i)

@network
pk : InputVector -> OutputVector

normpk : UnnormalisedInputVector -> OutputVector
normpk x = pk (normalise x)

@parameter
Ka, Ke, Vd, C_safe, ttd, Ka_over, Ka_under, Ke_over, Ka_under, eps : Real

@property
Ka_pos : Bool
Ka_pos = 0 < Ka

@property
Ke_pos : Bool
Ke_pos = 0 < Ke

@property
Ke_n_Ka : Bool
Ke_n_Ka = Ka != Ke

@property
Vd_pos : Bool
Vd_pos = 0 < Vd

@property
C_safe_pos : Bool
C_safe_pos = 0 < C_safe

@property
ttd_pos : Bool
ttd_pos = 0 < ttd

safeFarInput : InputVector -> Bool
safeFarInput x =
    0 <= x ! conc <= C_safe * 0.99 and
    36.5 <= x ! temp <= 40 and
    7.5 <= x ! wbc <= 20 and
    18 <= x ! age <= 89 and
    50 <= x ! weight <= 100

safeFarOutput : InputVector -> Bool
safeFarOutput x = let y = ((((normpk x) ! 0) * Ka) / (Vd * (Ka - Ke))) in
           if Ka < Ke
           then (x ! conc) + y * (Ke_under - Ka_over) < C_safe
           else (x ! conc) + y * (Ke_over - Ka_under) < C_safe

@property
safeFar : Bool
safeFar = forall x . safeFarInput x => safeFarOutput x

safeNearInput : InputVector -> Bool
safeNearInput x =
    C_safe * 0.99 <= x ! conc <= C_safe and
    36.5 <= x ! temp <= 40 and
    7.5 <= x ! wbc <= 20 and
    18 <= x ! age <= 89 and
    50 <= x ! weight <= 100

safeNearOutput : InputVector -> Bool
safeNearOutput x = ((normpk x) ! 0) < eps

@property
safeNear : Bool
safeNear = forall x . safeNearInput x => safeNearOutput x

safeInput : InputVector -> Bool
safeInput x =
    0 <= x ! conc <= C_safe and
    36.5 <= x ! temp <= 40 and
    7.5 <= x ! wbc <= 20 and
    18 <= x ! age <= 89 and
    50 <= x ! weight <= 100

nonNegOutput : InputVector -> Bool
nonNegOutput x =  0 < (normpk x) ! 0

@property
nonNeg : Bool
nonNeg = forall x . safeInput x => nonNegOutput x
\end{minted}
\section{Changes made}
\subsection{\vehicle changes}
The specification snippets given in~\autoref{sec:vehicle} are sufficient for the
proof, however, in~\aref{appendix:lst:vcl_spec}, there are more properties that
are slightly different. This is because $(\leq)$ gets translated into $(<)$ through
quantifier duality,
which the SMT-solver Marabou cannot handle, so the property is weakened to
$(\leq)$, which is unsound. Therefore, the specification is instead changed to
use $(<)$ instead of $(\leq)$, but this has consequences. This results in the
\texttt{nonNeg} property actually becoming a positivity property. Furthermore,
\texttt{safe}, as seen in~\autoref{sec:vehicle}, is split into two properties,
\texttt{safeFar} which is logically similar to \texttt{safe}, but has a slightly
reduced input domain and $(\leq)$ is weakened to $(<)$, and \texttt{safeNear}.
\texttt{safeNear} functions to fill the space left in the input domain of
\texttt{safeFar} and the entire range of valid $\csafe$ values. It restricts the
network to output a sufficiently small value, $\epsilon$. Anything less than
this $\epsilon$ value is considered to be 0. This means the properties are stronger
than we would ideally want, and slightly different, but it is provable
that \texttt{safeNear} and \texttt{safeFar} are sufficient for \texttt{safe}, and can be found in
\rocq in~\cite{Sirman_Vancomycert_Rocq_Files_2026}.

\bibliographystyle{splncs04}
\bibliography{bibliography}
\end{document}